\documentclass[runningheads,envcountsame,envcountsect,orivec]{llncs}
\usepackage[T1]{fontenc}
\usepackage{graphicx}
\usepackage{bbding} 
\usepackage{amssymb}
\usepackage{thm-restate}
\usepackage{ifthen}
\usepackage{url}
\usepackage{xcolor}
\usepackage{cleveref}
\usepackage{eclbkbox}

\usepackage{lineno}

\usepackage{latexsym}
\usepackage{amsfonts}
\usepackage{ifthen}

\def\cC{\mathcal{C}}
\def\cD{\mathcal{D}}
\def\cE{\mathcal{E}}

\def\cG{\mathcal{G}}
\def\cH{\mathcal{H}}
\def\cI{\mathcal{I}}
\def\cJ{\mathcal{J}}

\def\cR{\mathcal{R}}
\def\cS{\mathcal{S}}

\def\cV{\mathcal{V}}

\newcommand{\sort}[2][]{\ifthenelse{\equal{#1}{}}{\mbox{\sl{#2}}}{\mbox{\sl#1{#2}}}}

\newcommand{\SortBV}[1][]{\ifthenelse{\equal{#1}{}}{\sort{bv}}{\sort{bv}_{#1}}}
\newcommand{\sortBV}[1][]{\ifthenelse{\equal{#1}{}}{\sort[\scriptsize]{bv}}{\sort[\scriptsize]{bv}_{#1}}}

\newcommand{\symb}[1]{\mathsf{#1}}

\newcommand{\Pos}{\mathcal{P}os}
\newcommand{\Var}{\mathcal{V}ar}

\newcommand{\Int}{\mathbb{Z}}

\newcommand{\Dom}{\mathcal{D}om}
\newcommand{\Ran}{\mathcal{R}an}

\newcommand{\mgu}[1][]{\ifthenelse{\equal{#1}{}}{\mathit{mgu}}{\mathit{mgu}_{#1}}}

\newcommand{\cStheory}[1][\mathit{theory}]{\ifthenelse{\equal{#1}{}}{\cS}{\cS_{#1}}}
\newcommand{\cScore}{\cStheory[\mathit{core}]}

\newcommand{\cSterm}{\cS_{\mathit{term}}}

\newcommand{\Sigtheory}[1][\mathit{theory}]{\ifthenelse{\equal{#1}{}}{\Sigma}{\Sigma_{#1}}}
\newcommand{\Sigcore}{\Sigtheory[\mathit{core}]}
\newcommand{\Sigint}{\Sigtheory[{\sort[\scriptsize]{int}}]}

\newcommand{\Sigterm}{\Sigma_{\mathit{terms}}}
\newcommand{\Val}[1][]{%
\ifthenelse{\equal{#1}{}}{%
{\mathcal{V}al}%
}{%
{\mathcal{V}al}_{#1}%
}}

\newcommand{\Universe}{\mathcal{U}}
\newcommand{\Eval}[2][\cJ]{%
\ifthenelse{\equal{#1}{}}{%
{[\![ #2 ]\!]}
}{%
{[\![ #2 ]\!]}_{#1}
}}

\newcommand{\Constraint}[1]{[ #1 ]}
\newcommand{\CTerm}[2]{#1 \, \Constraint{#2}}

\newcommand{\LVar}{\mathcal{LV}ar}
\newcommand{\cRcalc}{\cR_{\mathit{calc}}}

\newcommand{\toBase}[1][base]{\to_{\mathsf{#1}}}

\newcommand{\NF}[1][]{\ifthenelse{\equal{#1}{}}{\mathit{NF}}{\mathit{NF\!}_{#1}}}
\newcommand{\GNF}[1][]{\ifthenelse{\equal{#1}{}}{\mathit{GNF}}{\mathit{GNF\!}_{#1}}}

\newcommand{\Rule}[2]{#1 \to #2}
\newcommand{\CRule}[3]{#1 \to #2 ~ \Constraint{#3}}

\newcommand{\simcto}{\stackrel{\smash{\raisebox{-.5mm}{\tiny $\sim$~\,}}}{\to}}
\newcommand{\cto}{\rightharpoonup}
\newcommand{\ctoBase}[1][\cR]{\to_{\mathsf{base},#1}}

\newcommand{\Eqn}[3][\approx]{#2 #1 #3}

\newcommand{\CEqn}[4][\approx]{\Eqn[#1]{#2}{#3}~\Constraint{#4}}

\newcommand{\ETerm}[3]{\exists #1: #2.\ #3}
\newcommand{\ETermBra}[3]{(\ETerm{#1}{#2}{#3})}
\newcommand{\CETerm}[4]{\ETermBra{#1}{#2}{#3} ~ \Constraint{#4}}

\newcommand{\EEqn}[5][\approx]{\exists #2: #3.\ #4 #1 #5}
\newcommand{\EEqnBra}[5][\approx]{(\EEqn[#1]{#2}{#3}{#4}{#5})}
\newcommand{\CEEqn}[6][\approx]{\EEqnBra[#1]{#2}{#3}{#4}{#5}~\Constraint{#6}}

\newcommand{\GInst}[2][]{\ifthenelse{\equal{#1}{}}{\ifthenelse{\equal{#2}{}}{\cG}{\cG(#2)}}{\ifthenelse{\equal{#2}{}}{\cG_{#1}}{\cG_{#1}(#2)}}}
\newcommand{\GInstNF}[1]{\ifthenelse{\equal{#1}{}}{\cG_{\mathrm{NF}}}{\cG_{\mathrm{NF}}(#1)}}

\newcommand{\RIstate}[2]{(#1, #2)}
\newcommand{\InferenceRule}[2]{\frac{~\displaystyle #1 ~}{~\displaystyle #2 ~}}
\newcommand{\Expd}[3][\cR]{\mathit{Expd}_{#1}(#2, ~ #3)}
\newcommand{\RIstep}[1][\mathrm{RI}]{\vdash_{#1}}
\def\Expansion{\textsc{Expansion}}
\def\Simplification{\textsc{Simplification}}
\def\CaseSplitting{\textsc{Case}}
\def\Deletion{\textsc{Deletion}}

\def\cRtoy{\cR_1}
\def\cRtrs{\cR_2}
\def\cRpow{\cR_3}

\renewcommand{\qed}{\hfill$\Box$}

\definecolor{mygray}{gray}{0.8}
\newcommand{\Comment}[4][mygray]{\ifthenelse{\equal{#3}{}}{}{\textcolor{#1}{#3}}\textcolor{#2}{#4}}

\begin{document}
\title{%
Rewriting Induction for Existentially Quantified Equations in
Logically Constrained Rewriting (Full Version)%
\thanks{This work was partially supported by 
Denso Corporation, NSITEXE,
JSPS KAKENHI Grant Number JP24K02900, 
Grant-in-Aid for JSPS Fellows Grant Number JP24KJ1240,
and
the Kayamori Foundation of Informational Science Advancement.
}%
\fnmsep%
\thanks{This work was mainly done as the master's study of the second author who graduated in March 2020.}
}
%
\titlerunning{Rewriting Induction for Existentially Quantified Equations in LCTRSs}
%
\author{%
Naoki Nishida%
\Envelope%
\and
Kazushi Nishie%
\and
Misaki Kojima%
\Envelope%
}
\authorrunning{N. Nishida, K. Nishie, and M. Kojima}
%
\institute{%
Graduate School of Informatics,
Nagoya University, \\
Furo-cho, Chikusa-ku, 
4648601 
Nagoya, Japan\\
\email{nishida@i.nagoya-u.ac.jp}
\quad
\email{kojima@i.nagoya-u.ac.jp}
}
\maketitle              
\begin{abstract}
Rewriting Induction (RI) is a principle to prove that an equation over terms is an inductive theorem of a rewrite system, i.e., that any ground instance of the equation is a theorem of the rewrite system.
RI has been adapted to several kinds of rewrite systems, and RI for constrained rewrite systems has been extended to inequalities.
In this paper, we extend RI for constrained equations to existentially quantified equations in logically constrained rewriting.
To this end, we first extend constrained equations by introducing existential quantification to the equation part of constrained equations.
Then, in applying a constrained rewrite rule to such extended constrained equations, we introduce existential quantification to extra variables of the applied rule.
Finally, using the extended application of constrained rewrite rules, we extend RI for constrained equations to existentially quantified equations.

\end{abstract}

\section{Introduction}
\label{sec:intro}

In the field of term rewriting, automated reasoning about \emph{inductive theorems} has been well investigated.
Here, an inductive theorem of a term rewrite system (TRS, for short) $\cR$ is an equation $\Eqn{s}{t}$ that is inductively valid, i.e., such that all of its ground instances $\Eqn{s\sigma}{t\sigma}$ are \emph{theorems} of $\cR$---$s\sigma\leftrightarrow_\cR^* t\sigma$.
As principles for proving inductive theorems, we cite \emph{inductionless induction}~\cite{Mus80,Hue82} and \emph{rewriting induction} (RI, for short)~\cite{Red90}, both of which are called \emph{implicit induction principles}.
Frameworks based on the RI principle (RI frameworks, for short) consist of inference rules to prove that given equations are inductive theorems.
RI-based methods are procedures within RI frameworks to apply inference rules under specified strategies or interactively.

In the last two decades, various RI-based methods for \emph{constrained rewriting}
have been developed~\cite{BJ08,SNSSK09,FK12,KN14aplas,FKN17tocl}.
Constrained rewrite systems
have built-in semantics for some function and predicate symbols
and can be computation models of not only functional but also imperative programs~\cite{FK08,FGPSF09,FNSKS08b,FK09,Vid12,KN14aplas,FKN17tocl,KNM25jlamp,CLB23}.
For this reason, it is worth developing verification techniques by means of constrained rewrite systems~\cite{FKN17tocl,WM18,CL18,NW18vstte,CLB23}.
The recent trend of constrained rewrite systems is the formalism of \emph{logically constrained term rewrite systems} (LCTRSs, for short)~\cite{KN13frocos}.
LCTRSs are used for several verification methods to ensure, e.g., 
termination~\cite{Kop13,FGK25}, 
equivalence~\cite{KN14aplas,FKN17tocl}, 
reachability~\cite{CL18,KN23jlamp},
runtime errors~\cite{KN23jlamp,KN23padl,KN24rp},
confluence~\cite{SM23,SMM24},
total correctness~\cite{BC18},
etc.
RI has been extended to LCTRSs~\cite{KN14aplas,FKN17tocl} and an LCTRS tool for RI-based equivalence verification of constrained equations has been developed~\cite{KN15lpar}.

To revisit the RI framework for LCTRSs, let us start with a toy example.
\begin{example}
\label{ex:toy}
Let us consider the following LCTRS over the integer signature:
\[
\cRtoy =
\{~~
\symb{f}(x,y) \to x + z ~ \Constraint{\, x > y \land z \geq x \,}, 
\quad
\symb{f}(x,y) \to x ~ \Constraint{\, x \leq y \,} 
~~~\}
\]
The constrained equation $\CEqn{\symb{f}(x,0)}{x+x}{x \geq 1}$ is an inductive theorem of $\cRtoy$ because
for any positive integer $n \geq 1$, 
we have that $\symb{f}(n,0) \to_{\cRtoy}^* n+n$ .
\end{example}

To prove a constrained equation $\CEqn{s}{t}{\phi}$ to be an inductive theorem of an LCTRS $\cR$,
we have to show that
$s\sigma \mathrel{\leftrightarrow_\cR^*} t\sigma$
for any ground substitution %
$\sigma$ satisfying $\phi$.
To be more precise, for each ground substitution $\sigma$, 
it suffices to show the existence of a sequence
$s\sigma \leftrightarrow_\cR u_1 \leftrightarrow_\cR \cdots \leftrightarrow_\cR u_n \leftrightarrow_\cR t\sigma$.
However, for $\CEqn{s}{t}{\phi}$, the inference rules {\Expansion} and {\CaseSplitting} of RI make an exhaustive case analysis based on some position of $s$ or $t$ by applying rewrite rules in $\cR$.
{\Expansion} makes an exhaustive case analysis by reducing the equation by all applicable rules at $p$ and by adding $\CRule{s}{t}{\phi}$ ($\CRule{t}{s}{\phi}$) to the set $\cH$ of constrained rewrite rules performing as \emph{induction hypotheses}, where 
$\cR\cup\cH\cup\{ \CRule{s}{t}{\phi} \}$ ($\cR\cup\cH\cup\{ \CRule{t}{s}{\phi} \}$) is terminating.
{\CaseSplitting} performs similarly to {\Expansion}, while it does not add any rules to $\cH$.
The application of {\Expansion} and {\CaseSplitting} is not complete in general.
\begin{example}
Let us consider $\cRtoy$ and $\CEqn{\symb{f}(x,0)}{x+x}{x \geq 1}$ in \Cref{ex:toy} again.
The root position of $\symb{f}(x,0)$ in the equation is \emph{reduction-complete} w.r.t.\ $\cRtoy$---all ground 
    instances are redexes of $\cRtoy$---and thus {\Expansion} and {\CaseSplitting} are applicable to $\CEqn{\symb{f}(x,0)}{x+x}{x \geq 1}$.
Both {\Expansion} and {\CaseSplitting} can convert 
the constrained equation to the following one:
\[
\CEqn{x + z}{x + x}{x \geq 1 \land z \geq x}
\]
The above equation is not an inductive theorem, because, e.g., for 
$\sigma = \{ x \mapsto 1, z \mapsto 2 \}$ satisfying $x \geq 1 \land z \geq x$, we have that 
$(x+z)\sigma \to_{\cRtoy}^* 3$
and
$(x+x)\sigma \to_{\cRtoy}^* 2$.
Note that $\cRtoy$ implicitly includes the \emph{calculation rule} $\CRule{x_1+x_2}{x}{x = x_1+x_2}$.
Note also that this problem does not originate from constrained rewriting.%
\footnote{%
The equation $\Eqn{\symb{f}(x)}{\symb{a}}$ is an inductive theorem of the TRS
$\cRtrs=\{
\Rule{\symb{f}(\symb{z})}{\symb{a}}, \,
\Rule{\symb{f}(\symb{z})}{\symb{b}}, \,
\Rule{\symb{f}(\symb{s}(x))}{\symb{a}}
\}$, 
not all equations 
$\Eqn{\symb{a}}{\symb{a}}, 
\Eqn{\symb{b}}{\symb{a}}$
obtained from $\Eqn{\symb{f}(x)}{\symb{a}}$ by applying {\Expansion} or {\CaseSplitting} are inductive theorems of $\cRtrs$.
}
For each positive integer $n \geq 1$,
the term $\symb{f}(n,0)$ has infinitely many reducts $n+k$ with $k\geq n$,
but for the equation being an inductive theorem, it suffices that
one of the reducts reaches $n+n$ by means of $\leftrightarrow_{\cRtoy}$.
On the other hand, {\Expansion} and {\CaseSplitting} generate 
the subgoal $\CEqn{x+z}{x+x}{x \geq 1 \land z \geq x}$ and, to show it an inductive theorem, we have to consider \emph{all} substitutions satisfying $x \geq 1 \land z \geq x$.
What we expect for the application of {\Expansion} and {\CaseSplitting} is to convert the constrained equation to 
the following problem:
\begin{quotation}\noindent
for any ground substitution $\sigma$ for $x$ with $x\sigma \geq 1$,
    is there \emph{some} $k$ such that $k \geq x\sigma$ and $(x\sigma + k) \leftrightarrow_{\cRtoy}^* (x\sigma + x\sigma)$.
\end{quotation}
To formulate the above problem as a constrained equation, 
the introduction of \emph{existential quantification} to constrained equations would be reasonable:
\[
(\exists z: z \geq x.\ x + z \approx x + x) ~ \Constraint{x \geq 1}
\]
\end{example}
Note that {\Simplification} has the similar problem:
if we use the \emph{most general constrained rewriting}~\cite{TSNA25ppdp} which is usually implemented in RI-based tools for LCTRSs, then we have the same problem;
otherwise, in applying a constrained rewrite rule to a constrained term, we have to choose appropriate values and/or variables for the extra variables (e.g., choose $x$ for $z$ above) of the applied rule.

In this paper, we introduce existential quantification to the equation part of constrained equations, formulating \emph{constrained existentially quantified equations} (constrained $\exists$-equations, for short) of the form
\[
\CEEqn{\vec{x}}{\eta}{s}{t}{\phi}
\]
in the LCTRS formalism (\Cref{sec:E-equations}).
Then, we extend the RI framework for LCTRSs to constrained $\exists$-equations (\Cref{sec:RI-for-constrained-E-equation}).
To simplify discussions, we restrict ourselves to left-value-free LCTRSs that are left-linear w.r.t.\ theory variables (\Cref{sec:restriction-to-LCTRSs}).
All omitted proofs of the claims are provided in both the appendix. 

In the extended RI framework, when we apply a constrained rewrite rule $\CRule{\ell}{r}{\varphi}$ to a constrained $\exists$-equation $\CEEqn{\vec{x}}{\eta}{s}{t}{\phi}$, we introduce existential quantification to the \emph{extra variables}---variables appear in $r$ but not in $\ell$---of the applied rule:
Roughly speaking, the constrained $\exists$-equation is reduced by the rule at a position $p$ of $s$ by means of a substitution $\gamma$ to
\[
\CEEqn{\vec{x},\vec{z}}{\eta \land \varphi\gamma}{s[r\gamma]_p}{t}{\phi}
\]
where $\vec{z}$ are the extra variables of the applied rule. 
While equations are existentially quantified, we do not extend constrained rewrite rules in LCTRSs.
In orienting a constrained $\exists$-equation $\CEEqn{\vec{x}}{\eta}{s}{t}{\phi}$ from left to right, 
we convert it to $\CRule{s}{t}{\phi \land \eta}$.

The key of extending RI to constrained $\exists$-equations must be the inference rule {\Deletion}.
In the RI framework for LCTRSs, {\Deletion} removes constrained equations $\CEqn{s}{t}{\phi}$ from the subgoals if
$s = t$ or $\phi$ is unsatisfiable.
Note that such equations are trivially inductive theorems.
In the extended RI framework, {\Deletion} removes constrained equations $\CEEqn{\vec{x}}{\eta}{s}{t}{\phi}$ from the subgoals if
$s = t$, $\phi$ is unsatisfiable, or
$s,t$ are theory terms such that $(\phi \Rightarrow (\exists \vec{x}.\ (\eta \land s = t))$ is valid.
Let $s,t$ be theory terms and $\gamma$ a substitution satisfying $\phi$, the domain of which does not include $\vec{x}$.
To use the third new condition,
we do not show validity but satisfiability of $s\gamma = t\gamma$.
This difference has a greater impact than we might expect.


A constrained object $\CTerm{P(\vec{y})}{\phi}$ represents the set of instances of $P(\vec{y})$ w.r.t.\ $\phi$:
$
\{ P(\vec{v}) \mid \mbox{$\vec{v}$ satisfy $\phi$} \}
$.
We consider it the conjunction of $P(\vec{v})$: $\bigwedge_{\mbox{\scriptsize $\vec{v}$ satisfy $\phi$}} P(\vec{v})$.
On the other hand, a constrained $\exists$-object $\CETerm{\vec{x}}{\eta}{P(\vec{y},\vec{x})}{\phi}$ represents 
the family of the sets of instances of $P(\vec{y},\vec{x})$ w.r.t.\ $\phi$:
$
\{\,
\{ P(\vec{v},\vec{w}) \mid \mbox{$\vec{v},\vec{w}$ satisfy $\eta$} \} 
\mid \mbox{$\vec{v}$ satisfy $\phi$}
\,\}
$.
The inner set can be considered the disjunction of $P(\vec{v},\vec{w})$, and
the whole as the conjunction of the disjunctions: $\bigwedge_{\mbox{\scriptsize $\vec{v}$ satisfy $\phi$}} \bigvee_{\mbox{\scriptsize $\vec{v},\vec{w}$ satisfy $\eta$}} P(\vec{v},\vec{w})$.
Viewed in this light, constrained $\exists$-objects are very similar to CNFs.

We have used a simple example (\Cref{ex:toy}) to clearly illustrate the key problem.
On the other hand, we have a practical motivation of introducing existential quantification to constrained equations, which is the reduction of constrained inequalities to constrained equations.

The RI framework for constrained TRSs in~\cite{SNSSK09} has been extended for constrained \emph{inequalities}~\cite{NN18scp}.
In the following, we describe the extended RI framework that are informally and naively adapted to LCTRSs.

Let us consider the following convergent (i.e., terminating and confluent) LCTRS over the integer signature (cf.\ \Cref{ex:pow} in \Cref{sec:preliminaries} for details):
\[
\cRpow =
\{~~
\CRule{\symb{pow}(x,y)}{1}{y < 1}, 
\quad
\CRule{\symb{pow}(x,y)}{x \times \symb{pow}(x,y-1)}{y \geq 1} 
~~\}
\]
The function symbol $\symb{pow}$ is defined by $\cRpow$ as the power function over the integers: for integers $m,n$, $\symb{pow}(m,n) \mathrel{\to_{\cRpow}^*} m^n$.
Let us try to prove that the following inequality
is an inductive theorem of $\cRpow$:
\begin{equation}
\label{eqn:pow-2-n-geq-n}
\CEqn[\geq]{\symb{pow}(2,n)}{n}{n \geq 0}
\end{equation}
In proving validity of constrained inequalities by means of the RI framework in~\cite{NN18scp},
the application of {\Simplification} w.r.t.\ oriented inequalities is in general \emph{not} complete (see \cite[Example~6.1]{NN18scp} for details).
If we found a subgoal is not an inductive theorem, then we backtrack and try to another application of inference rules.
Note that the method in~\cite{NN18scp} succeeds in proving the inequality~(\ref{eqn:pow-2-n-geq-n}) an inductive theorem of $\cRpow$.

In the case of RI frameworks for equations, the application of {\Simplification} is always complete because rules and induction hypotheses are oriented equations.
In~\cite[Section~1]{NN18scp}, to reduce inequalities to equations, the introduction of a binary symbol for $\geq$ is discussed.
For example, $(\ref{eqn:pow-2-n-geq-n})$ is reduced to
the constrained equation
$\CEqn{\symb{geq}(\symb{pow}(2,n),n)}{\symb{true}}{n \geq 0}$,
where rules $\CRule{\symb{geq}(x,y)}{\symb{true}}{x \geq y}$ and $\CRule{\symb{geq}(x,y)}{\symb{false}}{x \not\geq y}$ are introduced.
This approach has not been adopted in~\cite{NN18scp} because $\symb{geq}$ prevents us from simplifying proper subterms in the left-hand side in an expected way.

A naive way in first-order logic to reduce inequalities to equations is the use of \emph{existential quantification}.
For example, inequality $(\ref{eqn:pow-2-n-geq-n})$ is reduced to the following formula:
\begin{equation}
\label{eqn:naive-idea}
\CEEqn{m}{m \geq 0}{\symb{pow}(2,n)}{n+m}{n \geq 0}
\end{equation}
Unfortunately, RI frameworks cannot deal with such an equation. 
%
Our goal is to extend the RI framework for LCTRSs to constrained $\exists$-equations so as to prove that the constrained $\exists$-equation~(\ref{eqn:naive-idea}) above is an inductive theorem of $\cRpow$.

The contribution of this paper is 
to combine the RI framework for inequalities with that for equations by reducing inequalities to equations.
Inequalities are useful to verify some properties of functions;
for example, the returned values are always greater than or equal to an expression over input arguments, and the returned number of computation steps is always less than or equal to an estimated upper bound.
The RI framework for inequalities needs some special treatments of (oriented) inequalities, and {\Simplification} is in general not complete.
The simple combination of the two frameworks such that both inequalities and equations can be handled is more complicated than the framework for equations, and still has incompleteness of {\Simplification}.
On the other hand, the RI framework for constrained $\exists$-equations is a conservative extension of that for constrained equations, {\Simplification} is complete, and we can focus on only constrained ($\exists$-)equations.

\section{Preliminaries}
\label{sec:preliminaries}

In this section, we briefly recall LCTRSs~\cite{KN13frocos,FKN17tocl}.
Familiarity with basic notions and notations on term rewriting~\cite{BN98,Ohl02} is assumed. 

To define an LCTRS~\cite{KN13frocos,FKN17tocl}  over an $\cS$-sorted signature $\Sigma$, we consider the following sorts, signatures, mappings, and constants:
\emph{Theory sorts} in $\cStheory$ and \emph{term sorts} in $\cSterm$ such that $\cS=\cStheory \uplus \cSterm$;
a \emph{theory signature} $\Sigtheory$ and a \emph{term signature} $\Sigterm$ such that $\Sigma = \Sigtheory \cup \Sigterm$ and 
$\iota_1,\ldots,\iota_n,\iota\in \cStheory$ for any symbol $f: \iota_1 \times \cdots \times \iota_n \to \iota \in \Sigtheory$;
a mapping $\cI$ that assigns to each theory sort $\iota$ a (non-empty) set $\Universe_\iota$, so-called the universe of $\iota$ (i.e., $\cI(\iota) = \Universe_\iota$);
a mapping $\cJ$, so-called an interpretation for $\Sigtheory$, that assigns to each function symbol $f : \iota_1 \times \cdots \times \iota_n \to \iota \in \Sigtheory$ a function $f^\cJ$ in 
$\cI(\iota_1) \times \cdots \times \cI(\iota_n) \to \cI(\iota)$ (i.e., $\cJ(f)=f^\cJ$);
a set $\Val[\iota] \subseteq \Sigtheory$ of \emph{value-constants} $a: \iota$ for each theory sort $\iota$ 
such that $\cJ$ gives a bijection from $\Val[\iota]$ to $\cI(\iota)$.
We denote $\bigcup_{\iota \in \cStheory} \Val[\iota]$ by $\Val$.
Note that $\Val \subseteq \Sigtheory$.
For readability, we may not distinguish $\Val[\iota]$ and $\cI(\iota)$, i.e., for each $v \in \Val[\iota]$, $v$ and $\cJ(v)$ may be identified, calling them \emph{values}.
We require that $\Sigterm \cap \Sigtheory \subseteq \Val$.%
\footnote{%
In this paper, we may assume that $\Sigterm \cap \Sigtheory = \emptyset$.
}
Symbols in ${\Sigtheory} \setminus {\Val}$ are \emph{calculation symbols}, for which we may use infix notation.
A term in $T(\Sigtheory,\cV)$ is called a \emph{theory term}. 
We define the \emph{interpretation} $\Eval[\cJ]{\,\cdot\,}$ of ground theory terms as $\Eval[\cJ]{ f(s_1,\ldots,s_n) } = \cJ(f)(\Eval[\cJ]{ s_1 },\ldots,\Eval[\cJ]{ s_n })$.

We typically choose a theory signature $\cStheory$ such that 
$\cStheory \supseteq \cScore=\{\sort{bool}\}$,
$\Val[\scriptsize\sort{bool}]=\{ \symb{true}, \symb{false}: \sort{bool} \}$,
$\Sigtheory \supseteq \Sigcore=
\Val[\scriptsize\sort{bool}] \cup \{{\land}, {\vee}, {\Rightarrow}, {\Leftrightarrow} : \sort{bool} \times \sort{bool} \to \sort{bool}, ~\neg: \sort{bool} \to \sort{bool} \} \cup \{ {=_\iota}, {\ne_\iota} : \iota \times \iota \to \sort{bool} \mid 
\iota \in \cStheory
\}$, 
		and
$\cJ$ interprets these symbols as expected. 
We omit the sort subscripts $\iota$ from $=_\iota$ and $\neq_\iota$ when they are clear from the context.
A theory term with sort $bool$ is called a \emph{constraint}.
A substitution $\gamma$ is said to \emph{respect} a constraint $\phi$ if $x\gamma \in \Val$ for all $x \in \Var(\phi)$ and $\Eval[\cJ]{\phi\gamma} = \symb{true}$.
A constraint $\phi$ is said to be \emph{valid} (\emph{satisfiable}) if $\Eval{\phi\gamma} = \symb{true}$ for any (some) substitution $\gamma$ such that $\Dom(\gamma) \supseteq \Var(\phi)$ and $\Ran(\gamma|_{\Var(\phi)}) \subseteq \Val$.
We often consider validity of first-order formulas of the form $\phi \Rightarrow (\exists \vec{x}.\ \psi)$ which is not a constraint in our setting, where $\phi$ and $\psi$ are constraints (i.e., theory terms).
We define validity as follows:
$\phi \Rightarrow (\exists \vec{x}.\ \psi)$ is valid 
if and only if for any substitution $\gamma$ with $\Dom(\gamma)=\Var(\phi)\cup(\Var(\psi)\setminus\{\vec{x}\})$, there exists a substitution $\gamma'$ such that 
$\Dom(\gamma') = \{\vec{x}\}$,
$\Ran(\gamma') \subseteq \Val$,
and
$\Eval{\phi\gamma \Rightarrow \psi\gamma'\gamma} = \symb{true}$.
Note that validity of $(\exists \vec{x}.\ \phi) \Leftrightarrow (\exists \vec{y}.\ \psi)$ is defined by
the validity of $\phi \Rightarrow (\exists \vec{y}.\ \psi)$ and $\psi \Rightarrow (\exists \vec{x}.\ \phi)$, where $\phi$ and $\psi$ are constraints.

A \emph{constrained rewrite rule} is a triple $ \ell \to r ~[\varphi]$ such that $\ell$ and $r$ are terms of the same sort, $\varphi$ is a constraint, and 
$\ell$ is neither a theory term nor a variable.
If $\varphi = \symb{true}$, then we may write $\ell \to r$.
We define $\LVar(\ell \to r ~ [\varphi])$ as $\Var(\varphi) \cup (\Var(r) \setminus \Var(\ell))$,
the set of \emph{logical variables in $\ell \to r ~ [\varphi]$} which are variables instantiated with values in rewriting terms.
We say that a substitution $\gamma$ \emph{respects} $\ell \to r ~ [\varphi]$ if 
$\gamma(x) \in \Val$ for all logical variables $x \in \LVar(\ell \to r ~ [\varphi])$
and $\Eval[\cJ]{\varphi\gamma} = \symb{true}$.%
\footnote{Not all substitutions $\gamma$ respecting $\varphi$ respect $\ell \to r ~ [\varphi]$:
Such $\gamma$ ensures that $\gamma(x) \in \Val$ for all variables $x \in \Var(\varphi)$
but not that $\gamma(x) \in \Val$ for all $x \in \Var(r)\setminus \Var(\varphi)$.
}
We denote the set $\{ f(x_1,\ldots,x_n) \to y ~ [y = f(x_1,\ldots,x_n)] \mid f 
\in {{\Sigtheory} \setminus {\Val}}, ~
\mbox{$x_1,\ldots,x_n,y$ are pairwise distinct variables}
\}$ by $\cRcalc$.%
\footnote{Note that $\cRcalc$ is not determined by $\cR$ but by $\Sigtheory$.}
The elements of $\cRcalc$ are also called constrained rewrite rules (or \emph{calculation rules}) even though their left-hand sides are theory terms.

A constrained rewrite rule $\CRule{\ell}{r}{\varphi}$ is said to be \emph{left-value-free} 
if $\ell$ includes no value, i.e., $\ell \in T(\Sigma\setminus\Val,\cV)$~\cite{Kop17,TSNA25lopstr,TSNA25ppdp}.
The constrained rewrite rule $\CRule{\ell}{r}{\varphi}$ is said to be \emph{left-linear w.r.t.\ logical variables} (left-LV-linear, for short) if every logical variable appears in $\ell$ at most once~\cite{TSNA25lopstr,TSNA25ppdp}.
The constrained rewrite rule $\CRule{\ell}{r}{\varphi}$ is said to be \emph{left-linear w.r.t.\ theory variables} (left-TV-linear, for short) if every variable with a theory sort appears in $\ell$ at most once.
Note that any left-TV-linear constrained rule is left-LV-linear.

The \emph{rewrite relation} $\to_{\cR}$ of a set $\cR$ of constrained rewrite rules is a binary relation over terms, defined as follows:
For terms $s,t$,
$s \mathrel{\to_\cR} t$ if and only if 
there exist a constrained rewrite rule $\ell \to r ~ [\varphi] \in \cR \cup \cRcalc$,
a position $p$ of $s$, and a substitution $\gamma$
such that 
$s|_p = \ell\gamma$,
$t = s[r\gamma]_p$,
and
$\gamma$ respects $\ell \to r ~ [\varphi]$.
We may say that the reduction \emph{occurs at position $p$} and may write $\to_{p,\cR}$ or $\to_{p,\ell \to r ~[\varphi]}$ instead of $\to_\cR$.
A reduction step with $\cRcalc$ is called a \emph{calculation}.
A \emph{logically constrained term rewrite system} (LCTRS, for short) is defined as an abstract reduction system $(T(\Sigma,\cV),{\to_\cR})$, simply denoted by $\cR$. 
An LCTRS is usually given by supplying $\Sigtheory$, $\Sigterm$, $\cR$, and an informal description of $\cI$ and $\cJ$ if these are not clear from the context.
 The set of \emph{normal forms} of $\cR$ is denoted by $\NF_\cR$.
 The set of \emph{ground} normal forms of $\cR$ is denoted by $\GNF_\cR$.
The LCTRS $\cR$ is said to be \emph{terminating} if the rewrite relation $\to_\cR$ is terminating.

The \emph{standard integer signature} $\Sigint$ is $\Sigcore \cup \{ +, -,\times,\symb{exp},\symb{div},\symb{mod} : \sort{int} \times \sort{int} \to \sort{int} \} \cup \{ {\geq}, {>} : \sort{int} \times \sort{int} \to \sort{bool} \} \cup {\Val[\scriptsize\sort{int}]}$ where 
$\cStheory \supseteq \{\sort{int},\sort{bool}\}$,
$\Val[\scriptsize\sort{int}] = 
\mathbb{Z} 
$, $\cI(\sort{int}) = \mathbb{Z}$, and 
$\cJ(n) = n$ 
for any $n\in\mathbb{Z}$. 
We define $\cJ$ in a natural way.
An LCTRS 
with $\Sigtheory = \Sigint$ is called an \emph{integer LCTRS}.

\begin{example}
\label{ex:pow}
Let $\cS = \{ \sort{int},\sort{bool} \}$, and
$\Sigma = \Sigterm \cup \Sigtheory$, where
$
\Sigterm = \{~ \symb{pow} : \sort{int} \times \sort{int} \to \sort{int} ~\}$.
Then, both $\sort{int}$ and $\sort{bool}$ are theory sorts.
Examples of theory terms are $0 = 0+\symb{-1}$ and
$x+3 \geq y + -42$ which are constraints.
$5+9$ is also a (ground) theory term, but not a constraint.
$\symb{pow}(2,y)$ is not a theory term.
We reduce $3-1$ to $2$ in one step with the calculation rule $\CRule{x_1-x_2}{x}{x = x_1-x_2}$, and $3 \times (2 \times (1 \times 1))$ to $6$ in three steps. 
To implement an LCTRS calculating the \emph{power} function over $\Int$, we use the signature $\Sigma$ above 
 and 
 the LCTRS $\cRpow$ in \Cref{sec:intro}.
Using the constrained rewrite rules in $\cRpow$, $\symb{pow}(2,3)$ reduces in ten steps to $8$:
$
	\symb{pow}(2,3) 
	\mathrel{\to_{\cRpow}}
	 2 \times \symb{pow}(2,3-1)
	\mathrel{\to_{\cRpow}}
	 2 \times \symb{pow}(2,2)
	\mathrel{\to_{\cRpow}}
	\cdots 
	\mathrel{\to_{\cRpow}}
	 8
$.
\end{example}

A function symbol $f$ is called a \emph{defined symbol} of an LCTRS $\cR$ if there exists a rule $f(\ell_1,\ldots,\ell_n) \to r ~ [\varphi] \in \cR \cup \cRcalc$;
non-defined elements of $\Sigma$ are called \emph{constructors} of $\cR$.
Note that all values are constructors of $\cR$.
Note also that not all constructors are constants.
We denote the sets of defined symbols and constructors of $\cR$ by $\cD_\cR$ and $\cC_\cR$, respectively:
$\cD_\cR = \{ f \mid f(\ldots) \to r ~ [\varphi] \in \cR \cup \cRcalc \}$
and
$\cC_\cR = \Sigma \setminus \cD_\cR$.
A term in $T(\cC_\cR,\cV)$ 
is called a \emph{constructor term} of $\cR$.
A term of the form $f(t_1,\ldots,t_n)$ with $f\in\cD_\cR$ and constructor terms $t_1,\ldots,t_n$ is called \emph{basic}. 
The LCTRS $\cR$ is called \emph{quasi-reducible} if every ground basic term is a redex of $\cR$.
A method to prove quasi-reducibility of left-linear LCTRSs can be found in~\cite{NKN25frocos}.
When considering quasi-reducibility, we usually assume that for each sort $\iota$, there exists a constant with sort $\iota$.
The LCTRS $\cR$ is called \emph{sufficiently complete} if every ground term can be reduced to a ground constructor term.
Note that a terminating and quasi-reducible LCTRS is sufficiently complete.



A \emph{constrained term} is a pair $\CTerm{s}{\phi}$ of a term $s$ 
and a constraint $\phi$ 
with possibly $\Var(\phi) \not\subseteq \Var(s)$.
The constrained term $\CTerm{s}{\phi}$ is considered the set of all instances of $s$ w.r.t.\ $\phi$.
The set of ground instances of $s$ w.r.t.\ substitutions that respect $\phi$ is denoted by $\GInst{\CTerm{s}{\phi}}$:
$\GInst{\CTerm{s}{\phi}} = \{ s\gamma \in T(\Sigma) \mid \mbox{$\gamma$ is a substitution respecting $\phi$}\}$.%
\footnote{Since this paper deals with inductive theorems, we consider ground instances of constrained terms.}
Two constrained terms $\CTerm{s}{\phi}$ and $\CTerm{t}{\psi}$ are called \emph{unifiable} if 
$s$ and $t$ are unifiable with an mgu $\sigma$ such that
$\Ran(\sigma|_{\Var(\phi,\psi)}) \subseteq T(\Sigtheory,\cV)$
and $(\phi \land \psi)\sigma$ is satisfiable~\cite[Definition~4.4]{NKN25frocos} (cf.~\cite[Definition~2]{SM23}).

We say that constrained terms $\CTerm{s}{\phi}$ and $\CTerm{t}{\psi}$ are \emph{equivalent}, written as $\CTerm{s}{\phi} \sim \CTerm{t}{\psi}$, if for any substitution $\gamma$ respecting $\phi$, there exists a substitution $\delta$ respecting $\psi$ such that $s\gamma = t\delta$, and vice versa.
Note that $\CTerm{s}{\phi} \sim \CTerm{t}{\psi}$ if and only if 
$\GInst{\CTerm{s}{\phi}} = \GInst{\CTerm{t}{\psi}}$.
The \emph{rewrite relation} $\simcto_\cR$ over constrained terms is defined as follows:
$\CTerm{s}{\phi} \simcto_\cR \CTerm{t}{\psi}$ if and only if
$\CTerm{s}{\phi} \sim\cdot\ctoBase\cdot\sim \CTerm{t}{\psi}$, where 
${\ctoBase} = \{ (\CTerm{s'[\ell\gamma]_p}{\phi'}, \CTerm{s'[r\gamma]_p}{\phi'}) 
\mid \ell \to r ~ [\varphi] \in \cR\cup\cRcalc, ~ 
\Dom(\gamma) \supseteq \LVar(\ell \to r ~ [\varphi]), ~
\forall x \in \LVar(\ell \to r ~ [\varphi]).\ x\gamma \in \Val\cup\Var(\phi'), ~
\mbox{$(\phi \Rightarrow \varphi\gamma)$ is valid}\}$.
For $\simcto_\cR$ over constrained terms, we may say that the reduction \emph{occurs at position $p$} and may write $\simcto_{p,\cR}$ or $\simcto_{p,\ell \to r ~[\varphi]}$ instead of $\simcto_\cR$. 

\begin{example}
Consider the LCTRS $\cRpow$ in \Cref{ex:pow} again.
The constrained term $\CTerm{\symb{pow}(2,y)}{y > 2}$ represents the set $\{ \symb{pow}(2,3), \symb{pow}(2,4), \ldots \}$, and is reduced by $\cRpow$ as follows:
$\CTerm{\symb{pow}(2,y)}{y > 2}
\mathrel{\simcto_{\cRpow}}
\CTerm{2 \times \symb{pow}(2,y-1)}{y > 2}
\mathrel{\simcto_{\cRpow}}
\CTerm{2 \times (2 \times \symb{pow}(2,y'))}{y > 2 \land y' = y - 1}
\mathrel{\simcto_{\cRpow}}
\CTerm{2 \times (2 \times (2 \times \symb{pow}(2,y'')))}{y > 2 \land y'' = y' - 2}
$.
\end{example}

\section{Restriction to LCTRSs}
\label{sec:restriction-to-LCTRSs}

Constrained rewriting is defined with equivalence $\sim$ over constrained terms, which makes definitions, proofs, discussions, etc, very complicated.
To avoid the use of $\sim$, we restrict ourselves to some class of LCTRSs.
To be more precise, LCTRSs in this paper are assumed to satisfy all of the following:
\begin{enumerate}
\renewcommand{\theenumi}{\textbf{A\arabic{enumi}}}
\renewcommand{\labelenumi}{\textbf{A\arabic{enumi}}.}
\leftskip=2ex
    \item \label{A:no-extra-variables}
    for any constrained rewrite rule $\CRule{\ell}{r}{\varphi}$, all the extra variables---variables in $\Var(r) \setminus\Var(\ell)$---appear in $\varphi$: $\Var(r) \subseteq \Var(\ell,\varphi)$,
    \item \label{A:left-value-free}
    every constrained rewrite rule is left-value-free,
        and
    \item \label{A:left-TV-linear}
    every constrained rewrite rule is left-linear w.r.t.\ theory variables (left-TV-linear).
\end{enumerate}

The logical variables of $\CRule{\ell}{r}{\varphi}$ are instantiated with values in applying $\CRule{\ell}{r}{\varphi}$ to terms and 
are instantiated with values or logical variables in applying $\CRule{\ell}{r}{\varphi}$ to constrained terms.
Thus, to simplify discussions, we assume~{\ref{A:no-extra-variables}}:
For $\CRule{\ell}{r}{\varphi}$, if there is a variable $x$ in $\Var(r) \setminus \Var(\ell,\varphi)$, then we convert 
$\CRule{\ell}{r}{\varphi}$ to $\CRule{\ell}{r}{\varphi \land \bigwedge_{x \in \Var(r) \setminus \Var(\ell,\varphi)} (x = x)}$.

The assumption~{\ref{A:left-value-free}} is not an actual limitation because 
for any LCTRS $\cR$, there exists a left-value-free and left-LV-linear LCTRS $\cR'$ such that ${\to_{\cR}} = {\to_{\cR'}}$ on ground terms.
\begin{theorem}
Let $\CRule{\ell}{r}{\phi}$ be a constrained rewrite rule,  
$p_1,\ldots,p_n$ the positions 
with $\{p_1,\ldots,p_n\} = \{ p \in \Pos(\ell) \mid \ell|_p \in \Val \cup \Var(\phi) \}$,
and
$x_1,\ldots,x_n$ pairwise distinct variables such that $\{x_1,\ldots,x_n\} \cap \Var(\ell,r,\phi) = \emptyset$.
Let $\CRule{\ell'}{r}{\phi'}$ be a constrained rule such that 
    $\ell' = \ell[x_1,\ldots,x_n]_{p_1,\ldots,p_n}$
        and
    $\phi' = ( \phi \land (\bigwedge_{i=1}^n (x_i = \ell|_{p_i})) \land (\bigwedge_{y \in \Var(r)\setminus\Var(\ell,\varphi)} (y = y) )$.
Then, ${\to_{\CRule{\ell}{r}{\phi}}} = {\to_{\CRule{\ell'}{r}{\phi'}}}$.
\end{theorem}
\begin{proof}
Trivial by the definition of $\to_\cR$ and the construction of $\ell',\phi'$.
\qed
\end{proof}


The equivalence $\sim$ over constrained terms is used before $\ctoBase$
in order to, for a constrained term $\CTerm{s}{\phi}$, 
    rename variables in $s$ and $\phi$,
    move a value to the constrained part,
    conjunct a constraint $\phi'$ with $\phi$ for a constrained rewrite rule $\CRule{\ell}{r}{\varphi}$ applied at $\toBase$,
    etc.
The main role of $\sim$ for $\simcto_\cR$ is 
to make a rule $\CRule{\ell}{r}{\varphi}$ applicable at the $\ctoBase$-step.
Regarding left-value-free and left-TV-linear LCTRSs,
we have the following property.
\begin{restatable}{theorem}{ThemSimRemoval}
\label{thm:sim-removal}
Let $\CRule{\ell}{r}{\varphi}$ be a left-value-free and left-TV-free constrained rewrite rule in $\cR$,
$\CTerm{s}{\phi},\CTerm{t}{\psi}$ constrained terms, and $p$ a position of $s$ such that
$\Var(\ell,r,\varphi) \cap \Var(s,\phi,t,\psi) = \emptyset$.
If $\CTerm{s}{\phi} \sim \cdot 
\ctoBase[p,\CRule{\ell}{r}{\phi}] \CTerm{t}{\psi}$, then
there exists a substitution $\gamma$ such that
$\CTerm{s}{\phi} \sim \CTerm{s}{\phi \land \varphi\gamma} \ctoBase[p,\CRule{\ell}{r}{\varphi}] \CTerm{s[r\gamma]_p}{\phi \land \varphi\gamma}
\sim \CTerm{t}{\psi}$.
\end{restatable}

Thanks to \Cref{thm:sim-removal}, removing $\sim$ from $\simcto$, we reformulate constrained rewriting of 
LCTRSs satisfying~{\ref{A:no-extra-variables}--\ref{A:left-TV-linear}} as follows.
\begin{definition}[$\cto_\cR$]
\label{def:cto}
The constrained rewriting $\cto_\cR$ of a left-value-free and left-TV-linear LCTRS $\cR$ satisfying {\rm\ref{A:no-extra-variables}} is defined as follows:
$\CTerm{s}{\phi} \cto_\cR \CTerm{t}{\psi}$ if and only if
there exist a constrained rewrite rule $\CRule{\ell}{r}{\varphi} \in \cR\cup\cRcalc$, a position $p$ of $s$, and a substitution $\gamma$ such that 
\begin{itemize}
    \item $s|_p = \ell\gamma$,
    $t = s[r\gamma]_p$,
    $\psi = (\phi \land \varphi\gamma)$,
    \item 
    $x\gamma \in \Val\cup\Var(\phi)$ for any variable $x \in \Var(\varphi)\cap\Var(\ell)$,
    \item 
    $x\gamma \in \Val\cup\Var(\phi)\cup(\cV \setminus \Var(s))$ for any variable $x \in \Var(\varphi)\setminus\Var(\ell)$,%
        \footnote{In the definition of $\ctoBase$, for a variable $x \in \Var(\varphi) \setminus \Var(\ell)$, we can choose a value or a variable used in the constraint part.
    In the latter case, we can choose a fresh variable $y$ by adding a redundant constraint $y = y$.
    Note that $\Var(r) \subseteq \Var(\varphi,\ell)$ by~{\ref{A:no-extra-variables}}.
    }
        and
    \item 
    $(\phi \Rightarrow (\exists \vec{z}.\ \varphi\gamma))$ is valid, 
    where $\{\vec{z}\} = \{ z \in \cV \setminus \Var(\phi) \mid \exists x \in \Var(\varphi)\setminus\Var(\ell).\ x\gamma = z \}$.%
    \footnote{Since we can choose fresh variables for variables in $\Var(\varphi)\setminus\Var(\ell)$, 
    we existentially quantify such fresh variables $\vec{z}$ introduced by $\gamma$.}
\end{itemize}
\end{definition}
By definition, it is clear that ${\cto_\cR} \subseteq {\simcto_\cR}$.
Therefore, by \Cref{thm:sim-removal}, we have that
$({\cto_\cR\cdot\sim}) = {\simcto_\cR}$, and thus
$({\cto_\cR^n\cdot\sim}) = {\simcto_\cR^n}$ for all $n \geq 1$.

\section{Constrained $\exists$-Equations}
\label{sec:E-equations}

In this section, we first introduce existential quantification to the  equation part of constrained equations, formulating \emph{constrained $\exists$-equations}.
Then, we extend constrained rewriting to constrained $\exists$-equations.

\subsection{Introduction of Existential Quantification to Equations}
\label{subsec:eLCTRS}

First, we introduce \emph{existentially quantified equations}.

\begin{definition}[$\exists$-equation]
An \emph{existentially quantified equation} (\/$\exists$-equation, for short) is a quadruple
$\EEqn{\vec{x}}{\eta}{s}{t}$
of a finite variable sequence $\vec{x}$, a satisfiable constraint $\eta$, and terms $s,t$ of the same sort such that
    $\{\vec{x}\} \subseteq \Var(\eta) \cap \Var(s,t)$.
Note that if $\vec{x}$ is $\epsilon$, then $\eta$ is a closed valid formula, e.g., $\symb{true}$.
We may write $\Eqn{s}{t}$ instead of $\EEqn{\vec{x}}{\eta}{s}{t}$ if $\vec{x} = \epsilon$.
We say that $\EEqn{\vec{x}}{\eta}{s}{t}$ is \emph{closed} if $\Var(s,t,\eta) = \{\vec{x}\}$.
We denote the set of ground instances of $\EEqn{\vec{x}}{\eta}{s}{t}$ by 
$\GInst{\EEqn{\vec{x}}{\eta}{s}{t}}$:
$\GInst{\EEqn{\vec{x}}{\eta}{s}{t}}
=
\{ \Eqn{s\gamma}{t\gamma} \mid s\gamma,t\gamma \in T(\Sigma), ~ \mbox{$\gamma$ is a substitution respecting $\eta$} \}
$.
\end{definition}
We assume $\eta$ to be satisfiable so that $\EEqn{\vec{x}}{\eta}{s}{t}$ represents at least one equation.
When we have no constraint for the existentially quantified variables $\vec{x}$, we write a redundant constraint, e.g., $x = x$ in $\eta$.

\begin{example}
Consider the signature in \Cref{ex:pow}.
We have $\exists$-equations
$\EEqn{x}{x \geq 0}{\symb{pow}(2,y)}{x}$,
$\EEqn{x}{x = x}{\symb{pow}(2,x)}{2}$, etc.
\end{example}

Next, we introduce \emph{constrained $\exists$-equations}.

\begin{definition}[constrained $\exists$-equation]
A \emph{constrained $\exists$-equation} is a pair $\CEEqn{\vec{x}}{\eta}{s}{t}{\phi}$ of 
an $\exists$-equation $\EEqn{\vec{x}}{\eta}{s}{t}$ and a constraint $\phi$ such that
    $\Var(\eta) \setminus \{\vec{x}\} \subseteq \Var(\phi)$ 
	and
	$(\phi \Rightarrow (\exists \vec{x}.\ \eta))$ is valid.%
		\footnote{This ensures that for every substitution $\gamma$ respecting $\phi$, there exists at least one 
        equation $\Eqn{s\gamma}{t\gamma}$, i.e., $\GInst{\EEqn{\vec{x}}{\eta\gamma|_{\Var(\eta)\setminus\{\vec{x}\}}}{s\gamma}{t\gamma}} \ne \emptyset$.}
We simply write $\EEqn{\vec{x}}{\eta}{s}{t}$ if $\phi = \symb{true}$.
We denote the set of closed $\exists$-equations as instances of $\CEEqn{\vec{x}}{\eta}{s}{t}{\phi}$ w.r.t.\ 
substitutions that respect $\phi$ by 
$\GInst{\CEEqn{\vec{x}}{\eta}{s}{t}{\phi}}$: 
$\GInst{\CEEqn{\vec{x}}{\eta}{s}{t}{\phi}}
=
\{ \EEqn{\vec{x}}{\eta\gamma|_{\Var(\phi)\setminus\{\vec{x}\}}}{s\gamma}{t\gamma} \mid s\gamma,t\gamma \in T(\Sigma), ~ 
\mbox{$\gamma$ is a substitution respecting $\phi$} \}
$.
\end{definition}

\begin{example}
\label{ex:constrained-eterm}
Consider the signature in \Cref{ex:pow}.
We have constrained $\exists$-equations
$\CEEqn{x}{x \geq 2}{\symb{pow}(2,y)}{x}{y > 0}$,
$\CEEqn{x}{x > 0}{\symb{pow}(y,x)}{x}{y > 0}$, etc.
We also have that
$\GInst{\CEEqn{x}{x \geq 2}{\symb{pow}(2,y)}{x}{y > 0}} = 
\{\EEqn{x}{x \geq 2}{\symb{pow}(2,n)}{x} \mid n \in \mathbb{Z}, n > 0 \}$
and
$\GInst{\CEEqn{x}{x > 0}{\symb{pow}(y,x)}{x}{y > 0}} = 
\{ \EEqn{x}{x > 0}{\symb{pow}(n,x)}{x} \mid n \in \mathbb{Z}, n > 0 \}$.
\end{example}

\subsection{Constrained Rewriting of $\exists$-Terms}
\label{subsec:constrained-rewriting}

In the RI framework, by considering $\approx$ a binary function symbol, constrained equations are reduced by constrained rewrite rules.
In this section, considering equations as terms, we extend constrained rewriting to constrained $\exists$-equations.

To simplify discussions, we introduce (constrained) \emph{$\exists$-terms} as well as (constrained) $\exists$-equations.
An \emph{existentially quantified term} ($\exists$-term, for short) is a triple
$\ETerm{\vec{x}}{\eta}{s}$
of a finite variable sequence $\vec{x}$, a satisfiable constraint $\eta$, and a term $s$ such that
    $\{\vec{x}\} \subseteq \Var(\eta) \cap \Var(s)$.
We say that $\ETerm{\vec{x}}{\eta}{s}$ is \emph{closed} if $\Var(s) = \{\vec{x}\}$.
A \emph{constrained $\exists$-term} is a pair $\CETerm{\vec{x}}{\eta}{s}{\phi}$ of 
an $\exists$-term $\ETerm{\vec{x}}{\eta}{s}$ and a constraint $\phi$ such that
    $\Var(\eta) \setminus \{\vec{x}\} \subseteq \Var(\phi)$.
As for (constrained) $\exists$-equations, we use $\GInst{}$ 
for (constrained) $\exists$-terms.

As for constrained rewriting over usual constrained terms, 
we do not consider the equivalence between constrained $\exists$-terms.
Without using $\sim$, we define constrained rewriting of constrained $\exists$-terms as follows.

\begin{definition}[constrained rewriting $\cto_\cR$ of constrained $\exists$-terms]
\label{def:cto-for-constrained-E-terms}
The reduction $\cto_\cR$ of an LCTRS $\cR$ over constrained $\exists$-terms is defined as follows:
$\CETerm{\vec{x}}{\eta}{s}{\phi} \cto_{\cR} \CETerm{\vec{y}}{\delta}{t}{\psi}$
if and only if 
	there exist a 
    constrained rewrite rule $\CRule{\ell}{r}{\varphi} \in \cR\cup\cRcalc$, 
    a position $p$ of $s$, 
    and a substitution $\gamma$ such that
    \begin{itemize}
        \item $s|_p = \ell\gamma$,
        $t = s[r\gamma]_p$,
        $\psi = (\phi \land \varphi\gamma)$,
        \item 
        $x\gamma \in \Val\cup\Var(\phi)$ for any variable $x \in \Var(\varphi)\cap\Var(\ell)$,
        \item 
        $x\gamma \in \cV \setminus (\Var(s)\cup \{\vec{x}\})$ for any variable $x \in \Var(\varphi)\setminus\Var(\ell)$,%
        \footnote{For a variable $x \in \Var(\varphi)\setminus\Var(\ell)$, 
        we choose a fresh variable which is existentially quantified in the resulting constrained $\exists$-term.}
        \item
        $(\phi \Rightarrow (\exists \vec{z}.\ \varphi\gamma))$ is valid, 
            and
        \item 
        $(\exists \vec{y}.\ \delta) = 
        (\exists \vec{x},\vec{z}.\ (\eta \land \varphi\gamma))
        $, 
        where $\{\vec{z}\} = \{ x\gamma \mid \exists x \in \Var(\varphi)\setminus\Var(\ell) \}$.
    \end{itemize}
    In the above definition, we assume w.l.o.g.\ that $\{\vec{x}\} \cap \Var(s,\phi) = \emptyset$:
    if $\{\vec{x}\} \cap \Var(s,\phi) \ne \emptyset$, then we may rename $\vec{x}$ fresh variables.
\end{definition}
For readability, we often simplify $\psi$ and $\delta$ in \Cref{def:cto-for-constrained-E-terms} to \emph{equivalent} ones,%
\footnote{Constraints $\phi$ and $\phi'$ are \emph{equivalent}
if $\Var(\phi) = \Var(\phi')$ and $(\phi \Leftrightarrow \phi')$ is valid.}
respectively, without notice.
The differences of $\cto_\cR$ in \Cref{def:cto-for-constrained-E-terms} from $\cto_\cR$ in \Cref{def:cto} are the following:
\begin{itemize}
    \item variables in $\Var(\phi) \setminus \Var(\ell)$ are instantiated by $\gamma$ with fresh variables so as to be existentially quantified,
        and
    \item $\varphi\gamma$ is conjuncted with $\eta$ as $\exists \vec{x},\vec{z}.\ (\eta \land \varphi\gamma)$.    
\end{itemize}

\begin{example}
The constrained $\exists$-term $\CETerm{m}{m \geq 0}{(2\times \symb{pow}(2,n-1))+m}{n > 1}$ is reduced by $\cRpow$ as follows:
\[
\begin{array}{@{}l@{\>}c@{\>}l@{\hspace{-9ex}}l@{}}
\lefteqn{\CETerm{m}{m \geq 0}{(2\times \symb{pow}(2,n-1))+m}{n > 1}}
\\
& \cto_{\cRpow} &
\ETermBra{m,n'}{m \geq 0 \land n'=n-1}{(2\times \symb{pow}(2,n'))+m}&\Constraint{n > 1 \land n''=n-1}
\\
& \cto_{\cRpow} &
\ETermBra{m,n'}{m \geq 0 \land n'=n-1}{(2\times (2\times \symb{pow}(2,n'-1)))+m}\\&&&\Constraint{n > 1 \land n''=n-1}
\\
\end{array}
\]
Note that the variable $n''$ is the renamed one from $n'$ to avoid any shared variable between the constrained part and the constraint for existentially quantified variables.
\end{example}

Roughly speaking, to reduce a constrained term $\CTerm{s}{\phi}$ by a rule $\CRule{\ell}{r}{\varphi}$ at a position $p$ of $s$, 
\emph{all} ground terms in $\GInst{\CTerm{s}{\phi}}$ have to be reducible by $\CRule{\ell}{r}{\varphi}$ at $p$.
On the other hand, to reduce an $\exists$-term $\ETerm{\vec{x}}{\phi}{s}$ by $\CRule{\ell}{r}{\varphi}$ at $p$,
\emph{not all} ground terms in $\GInst{\ETerm{\vec{x}}{\phi}{s}}$ have to be reducible by $\CRule{\ell}{r}{\varphi}$ at $p$.
Note that $\GInst{\CTerm{s}{\phi}} = \GInst{\ETerm{\vec{x}}{\phi}{s}}$ by definition.
The same can be said of constrained $\exists$-terms.
\begin{example}
Let us consider $\cRtoy$ in \Cref{ex:toy} again.
The constrained term $\CTerm{\symb{f}(x,y)}{x \geq y \land y \geq 0}$ cannot be reduced by the first rule of $\cRtoy$ because $\symb{f}(0,0) \in \GInst{\CTerm{\symb{f}(x,0)}{x \geq 0}}$ is not reducible by the first rule.
On the other hand, the constrained $\exists$-term $\CETerm{x}{x \geq y}{\symb{f}(x,y)}{y \geq 0}$ is reduced by the first rule of $\cRtoy$ to $\CETerm{x,z}{x \geq y \land x > y \land z \geq x}{x + z}{y \geq 0}$, while $\symb{f}(0,0) \in \GInst{\CETerm{x}{x \geq y}{\symb{f}(x,y)}{y \geq 0}}$.
\end{example}

Constrained rewriting $\cto_\cR$ of constrained $\exists$-terms is sound for rewriting of terms, as well as usual constrained rewriting~\cite[Theorems~2.19 and 2.20]{FKN17tocl}.
%
\begin{restatable}{theorem}{ThmSoundnessOfConstrainedRewriting}
\label{thm:soundness_of_constrained-rewriting}
For constrained $\exists$-terms 
$\CETerm{\vec{x}}{\eta}{s}{\phi}$ and $\CETerm{\vec{y}}{\delta}{t}{\psi}$ with $\{\vec{x}\}\cap\Var(\phi)=\emptyset$ and $\{\vec{y}\}\cap\Var(\psi)=\emptyset$,
if
$\CETerm{\vec{x}}{\eta}{s}{\phi} \cto_\cR \CETerm{\vec{y}}{\delta}{t}{\psi}$,
then
both of the following statements hold:
\begin{enumerate}
	\item 
for every substitution $\theta$ respecting $\phi \land \eta$,
there exists a substitution $\sigma$ such that
$\sigma$ respects $\psi \land \delta$
and
$s\theta \to_\cR t\sigma$,
        and
	\item 
for every substitution $\sigma$ respecting $\psi \land \delta$,
there exists a substitution $\theta$ such that
$\theta$ respects $\phi \land \eta$
and
$s\theta \to_\cR t\sigma$.
\end{enumerate}
\end{restatable}

By definition, it is clear that all of the following are equivalent:
    $s \to_\cR t$,
    $\CTerm{s}{\symb{true}} \cto_\cR \CTerm{t}{\symb{true}}$,
        and
    $\CETerm{\epsilon}{\symb{true}}{s}{\symb{true}} \cto_\cR \CETerm{\epsilon}{\symb{true}}{t}{\symb{true}}$.
Thus, \Cref{thm:soundness_of_constrained-rewriting} implies the following corollary. 
\begin{corollary}
A left-value-free and left-TV-linear LCTRS $\cR$ satisfying~{\rm\ref{A:no-extra-variables}} is terminating if and only if $\cto_\cR$ over constrained $\exists$-terms is terminating.
\end{corollary}

\section{Rewriting Induction for Constrained $\exists$-Equations}
\label{sec:RI-for-constrained-E-equation}

In this section, we adapt the notion of \emph{inductive theorems} to constrained $\exists$-equations and then extend three main inference rules of the RI framework in~\cite{FKN17tocl} to constrained $\exists$-equations. 

When using RI frameworks, we often assume that a given system $\cR$ is 
	terminating,
	confluent,
	and
	quasi-reducible (i.e., sufficiently complete).
Termination is necessary to use RI principles, and confluence is necessary to disprove equations inductive theorems.
For page limitation, we do not deal with inference rules for disproving.
The reason why $\cR$ is assumed to be quasi-reducible is to use a decidable sufficient condition---being a \emph{basic} term---for \emph{reduction completeness} which ensures that, given a position of terms, we can make an exhaustive case analysis based on the application of rules in $\cR$.
%
In summary, we assume only termination of LCTRSs.


\subsection{Adapting Inductive Theorems to Constrained $\exists$-Equations}
\label{subsec:e-equations}



We naturally adapt the notion of inductive theorems to constrained $\exists$-equations.
\begin{definition}[inductive theorem]
A constrained $\exists$-equation 
$\CEEqn{\vec{x}}{\eta}{s}{t}{\phi}$
is said to be an \emph{inductive theorem} of $\cR$ if
for every ground normalized substitution $\gamma$ respecting $\phi$
such that $\Dom(\gamma) = \Var(s,t,\phi) \setminus \{\vec{x}\}$
and $\Ran(\gamma) \subseteq \GNF[\cR]$,
there exists a substitution $\theta$ such that
	$\Dom(\theta)=\{\vec{x}\}$,
	$\theta$ respects $\eta\gamma$,
	and
	$s\gamma\theta \mathrel{\leftrightarrow_\cR^*} t\gamma\theta$.
\end{definition}

\begin{example}
\label{ex:main-goal-equation}
The following constrained $\exists$-equation which corresponds to the inequality 
$(\ref{eqn:pow-2-n-geq-n})$ in \Cref{sec:intro}
is an inductive theorem of $\cRpow$:
\begin{equation}
\label{eqn:pow-initial-eqn}
	\CEEqn{m}{m \geq 0}{\symb{pow}(x,n)}{n + m}{x = 2 \land n \geq 0}
\end{equation}
where, to orient it, $2$ in $\symb{pow}(2,n)$ is replaced by a variable $x$ in advance, conjuncting $x=2$ with $n \geq 0$.
\end{example}

\subsection{RI Inference Rules for Constrained $\exists$-Equations}
\label{subsec:RI-inference-rules}

As described in \Cref{sec:intro}, RI frameworks usually consist of some inference rules.
Given a finite set $\cE$ of equations, we apply inference rules to the initial pair $\RIstate{\cE}{\emptyset}$ and try to get $\RIstate{\emptyset}{\cH}$ for some set $\cH$ of rewrite rules: 
\[
\RIstate{\cE}{\emptyset}
=
\RIstate{\cE_0}{\cH_0}
\mathrel{\RIstep}
\RIstate{\cE_1}{\cH_1}
\mathrel{\RIstep}
\cdots
\mathrel{\RIstep}
\RIstate{\cE_n}{\cH_n}
=
\RIstate{\emptyset}{\cH}
\]
where $\RIstep$ denotes the application of inference rules in one step.
A pair $\RIstate{\cE_i}{\cH_i}$ is called an \emph{RI state}.
$\cE_i$ is a set of equations and $\cH_i$ is a set of rewrite rules obtained by orienting equations.
Equations in $\cE_{i+1} \setminus \cE_i$ are so-called \emph{subgoals} for equations in $\cE_i\setminus \cE_{i+1}$.
Rewrite rules in $\cH_i$ perform as \emph{induction hypotheses} that can be used for $\RIstate{\cE_j}{\cH_j}$ with $i<j$.
The main inference rules are 
\begin{itemize}
	\item {\Expansion}---a rule-based case analysis with setting an induction hypothesis,
	\item {\Simplification}---the application of rules in $\cR\cup\cH_i$ to an equation in $\cE_i$,
		 and
	\item {\Deletion}---the removal of \emph{trivial} inductive theorems, e.g., $\Eqn{s}{s}$.
\end{itemize}
In this section, we adapt these three main inference rules for LCTRSs in~\cite{FKN17tocl} to constrained $\exists$-equations.
In the following, $\Eqn[\simeq]{s}{t}$ denotes either $\Eqn{s}{t}$ or $\Eqn{t}{s}$.

Before adapting the inference rules to constrained $\exists$-equations, we introduce the notion of \emph{reduction completeness} for positions of terms.
\begin{definition}[reduction-completeness]
A position $p$ of a term $t$ (and also $t|_p$) is said to be \emph{reduction-complete w.r.t.\ a constraint $\phi$} if $t|_p\gamma$ is a redex of $\cR$ for every substitution $\gamma$ such that $\Dom(\gamma)=\Var(t|_p)$, $\Ran(\gamma) \subseteq \NF_\cR(\Sigma)$, and $\gamma$ respects $\phi$.
\end{definition}
Roughly speaking, for a reduction-complete position $p$ of a term $t$, every ground instance of $t|_p$ can be reduced by $\cR$, and thus we can make an exhaustive rule-based case analysis at position $p$.
A naive sufficient condition for $f(t_1,\ldots,t_n)$ being reduction-complete w.r.t.\ a constraint $\phi$ is that 
$t_1,\ldots,t_n \in \Var(\phi)\cup\Val$ and 
$\bigvee_{\CRule{f(y_1,\ldots,y_n)}{r}{\varphi} \in \cR\cup\cRcalc} (\varphi\{y_i\mapsto t_i \mid 1 \leq i \leq n\})$ is valid (cf.~\cite[Theorem~5.1]{SNSSK09}).

The three main inference rules are adapted to constrained $\exists$-equations as shown in \Cref{fig:RI-inference-rules}.
The application of inference rules to $\RIstate{\cE}{\cH}$ with the resulting RI state $\RIstate{\cE'}{\cH'}$ is denoted by $\RIstate{\cE}{\cH} \mathrel{\RIstep} \RIstate{\cE'}{\cH'}$.

The application of the inference rules in \Cref{fig:RI-inference-rules} is sound, that is, if we obtain $\RIstate{\emptyset}{\cH}$ by applying the inference rules to the initial RI state $\RIstate{\cE}{\emptyset}$, then all constrained $\exists$-equations in $\cE$ are inductive theorems of $\cR$.
\begin{restatable}{theorem}{ThmSoundnessOfRI}
\label{thm:soundness_of_RI}
Given 
a set $\cE$ of constrained $\exists$-equations,
if an LCTRS $\cR$ is terminating and $\RIstate{\cE}{\emptyset} \mathrel{\RIstep^*} \RIstate{\emptyset}{\cH}$,
then all constrained $\exists$-equations in $\cE$ are inductive theorems of $\cR$.
\end{restatable}

\begin{figure}[tb]
\begin{breakbox}
\begin{description}
	\item[]\Simplification
	\[
	\InferenceRule{%
	\RIstate{\cE \mathop{\uplus} {\{~ 
    \CEEqn{\vec{x}}{\eta}{s}{t}{\phi}
    ~\}}}{\cH}
	}{
	\RIstate{{\cE} \cup {\{~
    \CEEqn{\vec{y}}{\eta'}{s'}{t'}{\phi'}
    ~\}}}{\cH}
	}
	\]
	if 
    $\CEEqn{\vec{x}}{\eta}{s}{t}{\phi} 
    \cto_{\cR\cup\cH} 
    \CEEqn{\vec{y}}{\eta'}{s'}{t'}{\phi'}
    $. 

	\medskip
	\item[]\Expansion
	\[
	\InferenceRule{%
	\RIstate{\cE \mathop{\uplus} {\{~
    \CEEqn[\simeq]{\vec{x}}{\eta}{s}{t}{\phi}
    ~\}}}{\cH}
	}{
	\RIstate{{\cE} \cup {\Expd{\CRule{s}{t}{\phi\land\eta}}{p}}}{\cH\cup\{\CRule{s}{t}{\phi\land\eta}\}}
	}
    \]
	if all of the following hold:
	\begin{itemize}
		\item $p$ is a reduction-complete position of $s$ w.r.t.\ $\phi$, 
		\item $\CRule{s}{t}{\phi\land\eta}$ is a constrained rewrite rule,
		and 
		\item $\cR \cup \cH \cup \{~ \CRule{s}{t}{\phi\land\eta} ~\}$ is terminating,
	\end{itemize}
	where $\Expd{\CRule{s}{t}{\phi\land\eta}}{p}$ denotes the following set that is finite up to variable renaming:
	\[
	\begin{array}{@{}lll@{}}
	\{&
    \CEEqn{\vec{y}}{\eta'}{s'}{t'}{\phi'}
	\mid
	\CRule{\ell}{r}{\varphi} \in \cR, ~ 
	\mbox{$\theta$ is an idempotent mgu of $s|_p$ and $\ell$}, 
    \\
    &
    \hspace{11ex}
    \CEEqn{\vec{x}}{\eta\theta|_{\Var(\phi)\setminus\{\vec{x}\}}}{s\theta}{t\theta}{\phi\theta}
    \cto_{1.p,\CRule{\ell}{r}{\varphi}}
    \CEEqn{\vec{y}}{\eta'}{s'}{t'}{\phi'}
	& \}
	\\
	\end{array}
	\]
	Note that $\CRule{\ell}{r}{\varphi}$ in the above description is freshly renamed:
    $\Var(\ell,r,\varphi) \cap \Var(s,t,\phi,\eta) = \emptyset$.

	\medskip
	\item[]\Deletion
	\[
	\InferenceRule{%
	\RIstate{\cE \mathop{\uplus} {\{~
    \CEEqn{\vec{x}}{\eta}{s}{t}{\phi}
    ~\}}}{\cH}
	}{
	\RIstate{\cE}{\cH}
	}
	\]
	if one of the following holds:
	\begin{itemize}
		\item 
        $s = t$,
		\item $\phi$ is unsatisfiable,
			or
		\item 
			$s,t \in T(\Sigtheory, \Var(\phi)\cup\{\vec{x}\})$
			and
			$(\phi \Rightarrow \exists \vec{x}.\ (\eta \land s = t))$ is valid.
	\end{itemize}
%
\end{description}
\end{breakbox}
\caption{Inference rules of RI for constrained $\exists$-equations.}
\label{fig:RI-inference-rules}
\end{figure}

\begin{example}
\label{ex:pow-RI-proof}
We now prove that the constrained $\exists$-equation (\ref{eqn:pow-initial-eqn}) in \Cref{ex:main-goal-equation} is an inductive theorem of the LCTRS $\cRpow$.
%
We start with the following RI state:
\[
\RIstate{\{ \,(\ref{eqn:pow-initial-eqn}) 
~
\CEEqn{m}{m \geq 0}{\symb{pow}(x,n)}{n + m}{x = 2 \land n \geq 0}
\,\}}{\emptyset}
\]
The 
left-hand side is reduction-complete w.r.t.\ $x = 2 \land n \geq 0$, and $\cRpow \cup \{\, \CRule{\symb{pow}(x,n)}{n + m}{x = 2 \land n \geq 0 \land m \geq 0} \,\}$ is terminating.
Thus, we can apply {\Expansion} to the above RI state, obtaining the following one:
\[
\begin{array}{@{}l@{}}
\RIstate{%
\left\{
\begin{array}{cr@{\>}c@{\>}ll}
(\mbox{c}) & 
\EEqnBra{m}{m \geq 0}{1 &}{& n+m} & \Constraint{x = 2 \land n \geq 0 \land n < 1}, \\
(\mbox{d}) &
\EEqnBra{m}{m \geq 0}{2\times \symb{pow}(2,n-1) &}{& n+m} & \Constraint{x = 2 \land n \geq 0 \land n \geq 1} \\
\end{array}
\right\}
}{%
\\[10pt]
\hspace{25ex}
\left\{
\begin{array}{cr@{\>}c@{\>}ll}
(\ref{eqn:pow-initial-eqn}') & \CRule{\symb{pow}(x,n) &}{& n + m &}{x = 2 \land n \geq 0 \land m \geq 0}
\\
\end{array}
\right\}
}
\end{array}
\]
Since $(x = 2 \land n \geq 0 \land n < 1 \Rightarrow (\exists m.\ m \geq 0 \land 1 = n + m))$ is valid, the constrained $\exists$-equation~(\mbox{c}) is dropped from the set by {\Deletion}, obtaining the following RI state:
\[
\RIstate{\{\,(\mbox{d})\,\}}{\{\, (\ref{eqn:pow-initial-eqn}') \,\}}
\]
Applying {\Simplification} to the above state, we reduce $n - 1$:
\[
\RIstate{%
\{\,
(\mbox{d}')~
\CEEqn{m,n'}{m \,{\geq}\, 0 \,\land\, n' \,{=}\, n{-}1}{2\times \symb{pow}(2,n')}{n+m}{x \,{=}\, 2 \,\land\, n \,{\geq}\, 1} 
\,\}
}{%
\{ (\ref{eqn:pow-initial-eqn}') \}}
\]
Applying {\Simplification} to the above state, we reduce $\symb{pow}(2,n')$ by $(\ref{eqn:pow-initial-eqn}')$, obtaining the following one:
\[
\begin{array}{@{}l@{}}
\RIstate{%
\left\{
\begin{array}{@{\,}l@{\,}}
(\mbox{d}'') ~
\EEqnBra{m,n',m'}{m \,{\geq}\, 0 \,\land\, n' \,{=}\, n{-}1 \,\land\, m'\,{\geq}\,0}{2\times (n'{+}m')}{n+m}\\
\hspace{57ex}
\Constraint{x \,{=}\, 2 \,\land\, n \,{\geq}\, 1} 
\\
\end{array}
\right\}
}{%
\\
\hspace{72ex}
\{\, (\ref{eqn:pow-initial-eqn}') \,\}
}
\end{array}
\]
We can drop $(\mbox{d}'')$ by {\Deletion}, obtaining the RI state
$
\RIstate{\emptyset}{\{\, (\ref{eqn:pow-initial-eqn}') \,\}}
$.
Therefore, the constrained $\exists$-equation $(\ref{eqn:pow-initial-eqn})$ is an inductive theorem of $\cRpow$.

As described in \Cref{sec:intro}, the RI framework for inequalities can prove inequality $(\ref{eqn:pow-2-n-geq-n})$ an inductive theorem of $\cRpow$.
The RI framework for constrained $\exists$-equations can trace the proof without any incomplete application of the RI inference rules.
\end{example}


\section{Related Work}
\label{sec:related-work}

The main achievement of this paper is to introduce existential quantification to constrained equations, and to extend the RI framework for LCTRSs to constrained $\exists$-equations.
To the best of our knowledge, there is no work on the introduction of existential quantification to equations in constrained rewriting.
Existentially constrained terms of the form $\CTerm{\mathrm{\Pi} X.\ s}{\exists \vec{x}.\ \phi}$ introduced in~\cite{TSNA25lopstr} use existential quantification for variables in $\Var(\phi) \setminus \Var(s)$, where $X$ is used to make logical variables explicit.
The role of the existential quantification is to make variables in $\Var(\phi) \setminus \Var(s)$ explicit and to have that $\Var(\phi) \setminus \{\vec{x}\} \subseteq \Var(s)$.
For this reason, the use of existential quantification is quite different from ours.
The notion of \emph{guarded terms}~\cite{BR17} is a generalization of constrained terms, but does not include $\exists$-terms.
A guarded term represents the set of all its instances, but an $\exists$-term represents \emph{one} of its instances.

As described in \Cref{ex:pow-RI-proof}, the RI framework for constrained $\exists$-equations would cover the RI framework for inequalities~\cite{NN18scp} because every inequality can be transformed into a semantically equivalent $\exists$-equation.

\section{Conclusion}
\label{sec:conclusion}

In this paper, we introduced existential quantification to the equation part of constrained equations, formulating constrained $\exists$-equations.
We defined the reduction over constrained $\exists$-equations ($\exists$-terms) and then extended the RI framework for LCTRSs to constrained $\exists$-equations.

We have implemented a prototype of the RI framework for constrained $\exists$-equations, but the prototype has not been public yet.
The prototype does not have a powerful method for proving termination of LCTRSs, and has only an interactive mode for proving inductive theorems by means of the extended RI framework.
We confirmed that the inference in \Cref{ex:pow-RI-proof} can be done by the prototype without proving termination, and termination of $\cRpow\cup \{\, (\ref{eqn:pow-initial-eqn}')\,\}$ can be proved by \textsf{Ctrl}~\cite{KN15lpar}, a tool for verifying LCTRSs.
We will automate the application of inference rules and will release the prototype.

In the extend RI framework, we consider only the main three inference rules, while there are other useful inference rules for LCTRSs.
The extension of such excluded inference rules to constrained $\exists$-equations is left as future work.

%
%
%
\bibliographystyle{splncs04}
\bibliography{mybiblio}

\newpage
\appendix

\section{Omitted Proofs}
\label{sec:proofs}

\subsection{Proof of \Cref{thm:sim-removal}}

\ThemSimRemoval*
\begin{proof}
By definition, we have that $\phi' = \psi$.
Since $\CTerm{s'}{\phi'} \ctoBase[p,\CRule{\ell}{r}{\phi}] \CTerm{t}{\psi}$,
there exists a substitution $\sigma$ such that 
$s'|_p = \ell\sigma$,
$t = s'[r\sigma]_p$,
and
$(\phi' \Rightarrow \varphi\sigma)$ is valid.
For every constrained term, there exists an equivalent value-free and LV-linear constrained term~\cite{KN24jip,TSNA25lopstr}.
Let $\CTerm{s_0}{\phi_0}$ be an equivalent value-free and LV-linear constrained term of $\CTerm{s}{\phi}$: $\CTerm{s_0}{\phi_0} \sim \CTerm{s}{\phi}$.
Then, there exists a substitution $\theta$ such that 
$x\theta \in \Val \cup \Var(\phi)$ for any variable $x \in \Var(s_0) \cap \Var(\phi_0)$,
$s=s_0\theta$, $\Var(\phi) = \Var(\phi_0\theta)$, and $(\phi \Leftrightarrow \phi_0\theta)$ is valid.
In addition, since $\CTerm{s}{\phi} \sim \CTerm{s'}{\phi'}$, there exists a substitution $\theta'$ such that 
$\Dom(\theta') = \Var(s_0)\cap\Var(\phi)$,
$x\theta' \in \Val \cup \Var(\phi')$ for any variable $x \in \Var(s_0) \cap \Var(\phi_0)$,
$s'=s_0\theta'$, $\Var(\phi') = \Var(\phi_0\theta')$, and $(\phi' \Leftrightarrow \phi_0\theta')$ is valid.
Since $\ell$ is value-free and TV-linear and matches $s'|_p$,
there exists a substitution $\sigma_0$ such that $s_0|_p = \ell\sigma_0$.
Then, we have that 
$s|_p = (s_0|_p)\theta = \ell\sigma_0\theta$,
$(\phi \land \varphi\sigma_0\theta \Rightarrow \varphi\sigma_0\theta)$ is valid,
and thus
$\CTerm{s}{\phi\land \varphi\gamma} = \CTerm{s}{\phi\land \varphi\sigma_0\theta} \ctoBase[p,\CRule{\ell}{r}{\varphi}] \CTerm{s[r\sigma_0\theta]_p}{\phi\land \varphi\sigma_0\theta} = \CTerm{s[r\gamma]_p}{\phi\land\varphi\gamma} \sim \CTerm{t}{\psi}$.
\qed
\end{proof}

\subsection{Proof of \Cref{thm:soundness_of_constrained-rewriting}}

\ThmSoundnessOfConstrainedRewriting*
\begin{proof}
Suppose that $\CETerm{\vec{x}}{\eta}{s}{\phi} \cto_{\cR} \CETerm{\vec{y}}{\delta}{t}{\psi}$.
We assume w.l.o.g.\ that $\{\vec{x}\} \cap \Var(s,\phi) = \emptyset$.
Then, there exist a constrained rewrite rule $\CRule{\ell}{r}{\varphi} \in \cR\cup\cRcalc$, 
    a position $p$ of $s$, and a substitution $\gamma$ such that
    \begin{itemize}
        \item $s|_p = \ell\gamma$,
        $t = s[r\gamma]_p$,
        $\psi = (\phi \land (\varphi\gamma))$,
        \item 
        $x\gamma \in \Val\cup\Var(\phi)$ for any variable $x \in \Var(\varphi)\cap\Var(\ell)$,
        \item 
        $x\gamma \in \cV \setminus (\Var(s)\cup \{\vec{x}\})$ for any variable $x \in \Var(\varphi)\setminus\Var(\ell)$,
        \item
        $(\phi \Rightarrow (\exists \vec{z}.\ \varphi\gamma))$ is valid, 
            and
        \item 
        $(\exists \vec{y}.\ \delta) = 
        (\exists \vec{x},\vec{z}.\ (\eta \land \varphi\gamma))
        $, 
        where $\{\vec{z}\} = \{ x\gamma \mid \exists x \in \Var(\varphi)\setminus\Var(\ell) \}$.
    \end{itemize}

We first prove the first claim.
Let $\theta$ be a substitution respecting $\phi \land \eta$.
Since $(\phi \Rightarrow (\exists \vec{z}.\ \varphi\gamma))$ is valid,
there exists a substitution $\theta'$ such that 
$\Dom(\theta') ~ \{ \vec{z} \}$,
$\Ran(\theta') \subseteq \Val$,
and
$\Eval{(\varphi\gamma)(\theta\cup\theta')}=\symb{true}$.
Let $\sigma=(\theta\cup\theta')$.
Then, $\sigma$ is a substitution respecting $\psi\land \delta= (\phi \land (\varphi\gamma) \land \eta)$,
and thus we have that 
$s\theta = (s[\ell\gamma]_p)\theta \to_\cR (s[r\gamma]_p)(\theta\cup\theta') = t\sigma$.

The second claim can be proved similarly to the first one.
\qed
\end{proof}

\subsection{Proof of \Cref{thm:soundness_of_RI}}

The proof of \Cref{thm:soundness_of_RI} follows~\cite[Section~4.4]{FKN17tocl}.
For a set $\cE$ of constrained $\exists$-equations,
we denote the set $\{ \CRule{s}{t}{\phi\land\eta} \mid \CEEqn{\vec{x}}{\eta}{s}{t}{\phi} \in \cE \}$ by $\tilde{\cE}$.
The oriented equation $\CRule{s}{t}{\phi\land\eta}$ may not be a constrained rewrite rule, but we deal with it as a constrained rewrite rule below.

\begin{definition}[$\leftrightarrow_\cE$ over ground terms]
Let $\cE$ be a set of constrained $\exists$-equations.
The relation $\leftrightarrow_\cE$ over ground terms is defined as follows:
for ground terms $s,t$,
$s \mathrel{\leftrightarrow_\cE} t$
if and only if
$s \mathrel{\leftrightarrow_{\tilde{\cE}}} t$.
\end{definition}

\begin{lemma}
\label{lem:RIstep-properties}
If $\RIstate{\cE}{\cH}	 \mathrel{\RIstep} \cdots \mathrel{\RIstep} \RIstate{\cE'}{\cH'}$,
then all of the following hold:
\begin{enumerate}
	\item 
$
{\leftrightarrow_\cE} 
\subseteq
{\left(
{\to_{\cR\cup\cH'}^*}
\cdot
{\leftrightarrow_{\cE'}^=}
\cdot
{\gets_{\cR\cup\cH'}^*}
\right)}
$
over ground terms,
	\item if
$
{\to_{\cR\cup\cH}}
\subseteq
{(
{\to_\cR}
\cdot
{\to_{\cR\cup\cH}^*}
\cdot
{\leftrightarrow_{\cE}^=}
\cdot
{\gets_{\cR\cup\cH}^*}
)}
$
over ground terms,
then
$
{\to_{\cR\cup\cH'}}
\subseteq
{(
{\to_\cR}
\cdot
{\to_{\cR\cup\cH'}^*}
\cdot
{\leftrightarrow_{\cE'}^=}
\cdot
{\gets_{\cR\cup\cH'}^*}
)}
$
over ground terms,
        and
	\item if $\cR\cup\cH$ is terminating, then $\cR\cup\cH'$ is so.
\end{enumerate}
\end{lemma}
\begin{proof}
It suffices to show the case where $\RIstate{\cE}{\cH} \mathrel{\RIstep} \RIstate{\cE'}{\cH'}$.
\begin{enumerate}
	\item 
		Since the case where the applied equation is in $\cE'$ is trivial, we consider the remaining case where the applied equation is in $\cE\setminus\cE'$.
	We make a case analysis depending on which inference rule is applied.
	\begin{itemize}
		\item Case where {\Simplification} is applied.
			Assume that the applied equation is 
            $
            \CEEqn[\simeq]{\vec{x}}{\eta}{s}{t}{\phi}
            \in \cE\setminus\cE'$, 
			where
            \begin{itemize}
            \item
			$
            \CEEqn{\vec{x}}{\eta}{s}{t}{\phi}
            \cto_{\cR\cup\cH}
            \CEEqn{\vec{y}}{\eta'}{s'}{t'}{\phi'}
			$,
			and
			\item
            $
            \CEEqn{\vec{y}}{\eta'}{s'}{t'}{\phi'}
			\in \cE'
			$.
            \end{itemize}
			It follows from \Cref{thm:soundness_of_constrained-rewriting} that
			$
            {\leftrightarrow_\cE}
            \subseteq
            {\to_{\cR\cup\cH}}
            \cdot
            {\leftrightarrow_{\cE'}}
			$.
			Therefore, the claim holds.

		\item Case where {\Expansion} is applied.
			Assume that the applied equation is 
            $
            \CEEqn[\simeq]{\vec{x}}{\eta}{s}{t}{\phi}
            \in \cE\setminus\cE'$,
			where
				\begin{itemize}
                \item $\CRule{s}{t}{\phi\land\eta}$ is a constrained rewrite rule,
		\item $\cE' \setminus \cE = 
    \Expd{\CRule{s}{t}{\phi\land\eta}}{p}
$,
		\item $\cH'=\cH\cup \{
        \CRule{s}{t}{\phi\land\eta}
        \}$,
		\item $p$ is a reduction-complete position of $s$ under $\phi$,
			and
		\item $\cR \cup \cH'$ is terminating.
	\end{itemize}
	Since $p$ is a reduction-complete position of $s$,
	there exist a constrained rewrite rule $\CRule{\ell}{r}{\varphi} \in \cR$ and a substitution $\theta$ such that
	\begin{itemize}
		\item $\theta$ is more general than $\gamma$,
		\item $\theta$ is an idempotent mgu of $s|_p$ and $\ell$,
		and
		\item 
        $
        \CEEqn{\vec{x}}{\eta\theta}{s\theta}{t\theta}{\phi\theta}
        \cto_{p,\CRule{\ell}{r}{\varphi}}
        \CEEqn{\vec{y}}{\eta'}{s'}{t'}{\phi'}
        \in \Expd{\CRule{s}{t}{\phi\land\eta}}{p}$
	\end{itemize}
    Thus, we have that
    $
            {\leftrightarrow_{\{ \CEEqn{\vec{x}}{\eta}{s}{t}{\phi} \}}}
            \subseteq
            {\to_{\cR}}
            \cdot
            {\leftrightarrow_{\cE'}}
	$.
	Therefore, this claim holds.

		\item Case where {\Deletion} is applied.
		Assume that the applied equation is 
        $
        \CEEqn{\vec{x}}{\eta}{s}{t}{\phi}
        \in \cE\setminus\cE'$,
		where if one of the following holds:
		\begin{itemize}
			\item $s = t$,
			\item $\phi$ is unsatisfiable,
				or
			\item 
            $s,t \in T(\Sigtheory, \Var(\phi)\cup\{\vec{x}\})$
			and
			$(\phi \Rightarrow \exists \vec{x}.\ (\eta \land s = t))$ is valid.
		\end{itemize}
		In the first case above, the claim trivially holds.
		In the second case, the equation is not instantiated.
		We now consider the third case.
        Let $\gamma$ be a ground substitution respecting $\phi$.
        Since $\Eval{\phi\gamma}=\symb{true}$,
        there exists a substitution $\theta$ respecting $(\eta\land(s=t))\gamma|_{\Var(\phi)\setminus\{\vec{x}\}}$.
        Thus, we have that $\Eval{s\gamma|_{\Var(\phi)\setminus\{\vec{x}\}}\theta}=
        \Eval{t\gamma|_{\Var(\phi)\setminus\{\vec{x}\}}\theta}$, and thus
        $s\gamma|_{\Var(\phi)\setminus\{\vec{x}\}}\theta
        \mathrel{\leftrightarrow_\cR^*}
        t\gamma|_{\Var(\phi)\setminus\{\vec{x}\}}\theta$.
		Therefore, the claim holds.
	\end{itemize}
	\item Thanks to the assumption, it suffices to consider the case where $\cH \subset \cH'$, i.e.,
    {\Expansion} is applied.
        This claim can be proved similarly to the case of applying {\Expansion} in the first claim.
	\item Trivial by definition.
\qed
\end{enumerate}
\end{proof}

\begin{lemma}[a variant of RI principle~{\cite[Theorem~A.1.2]{SNSSK09}}]
\label{lem:RI-principle}
Let ARSs $(A,\to_1)$ and $(A,\to_2)$.
Assume that 
\begin{itemize}
	\item ${\to_1} \subseteq {\to_2}$,
	\item $\to_2$ is terminating,
	and
	\item ${\to_2} \subseteq {({\to_1}\circ{\to_2^*}\circ {\leftrightarrow_1^*} \circ{\gets_2^*})}$.
\end{itemize}	
Then, ${\leftrightarrow_1^*} = {\leftrightarrow_2^*}$.
\end{lemma}

\ThmSoundnessOfRI*
\begin{proof}
Let 
$\CEEqn{\vec{x}}{\eta}{s}{t}{\phi} \in \cE$,
$\gamma$ a ground substitution respecting $\phi$,
and
$\theta$ a ground substitution respecting $\eta\gamma|_{\Var(\phi)\setminus\{\vec{x}\}}$.
We show that $s\eta\gamma|_{\Var(\phi)\setminus\{\vec{x}\}}
\mathrel{\leftrightarrow_\cR^*}
t\eta\gamma|_{\Var(\phi)\setminus\{\vec{x}\}}$.
It follows from \Cref{lem:RIstep-properties} that
\[
s\eta\gamma|_{\Var(\phi)\setminus\{\vec{x}\}}
\mathrel{%
{\left(
{\to_{\cR\cup\cH}^*}
\cdot
{\gets_{\cR\cup\cH}^*}
\right)}
}
t\eta\gamma|_{\Var(\phi)\setminus\{\vec{x}\}}
\]
It follows from \Cref{lem:RI-principle} that
${\leftrightarrow_\cR^*} = {\leftrightarrow_{\cR\cup\cH}^*}$
over ground terms,
and hence
\[
s\eta\gamma|_{\Var(\phi)\setminus\{\vec{x}\}}
\mathrel{\leftrightarrow_\cR^*}
t\eta\gamma|_{\Var(\phi)\setminus\{\vec{x}\}}
\]
Therefore, $\CEEqn{\vec{x}}{\eta}{s}{t}{\phi} \in \cE$ is an inductive theorem of $\cR$.
\qed	
\end{proof}

\end{document}